\documentclass[a4paper,USenglish,thm-restate]{lipics-v2021}

\pdfoutput=1
\hideLIPIcs
\nolinenumbers
\usepackage{cite}

\title{Towards an HRS Category in TermCOMP}

\author{Johannes Niederhauser}{University of Innsbruck, Innsbruck, Austria}
{johannes.niederhauser@uibk.ac.at}{https://orcid.org/0000-0002-8662-6834}{}
\author{Aart Middeldorp}{University of Innsbruck, Innsbruck, Austria}
{aart.middeldorp@uibk.ac.at}{https://orcid.org/0000-0001-7366-8464}{}

\authorrunning{J. Niederhauser and A. Middeldorp} 

\Copyright{Johannes Niederhauser and Aart Middeldorp}

\ccsdesc{Theory of computation~Equational logic and rewriting}
\ccsdesc{Theory of computation~Lambda calculus}

\keywords{higher-order rewriting, termination}

\acknowledgements{We thank Cynthia Kop as well as the anonymous reviewers
for their valuable and extensive feedback.}

\usepackage{definitions}

\begin{document}

\maketitle

\begin{abstract}
We show that there is a simple syntactically-defined subclass of
higher-order benchmarks in the termination problem database for
which rewriting according to Nipkow's higher-order rewrite systems
(HRSs) and rewriting according to a beta-first strategy in the
semantics of TermCOMP's higher-order category coincide. This lays
the formal foundation for an HRS (sub)category in TermCOMP which
would allow more tools to compete against each other.
\end{abstract}

\section{Introduction}

Higher-order rewriting comes in many different flavors. Popular
choices include Jouannaud and Okada's algebraic functional systems
(AFSs) \cite{JO91,JR99} and higher-order rewrite systems (HRSs) as
defined by Wolfram and Nipkow \cite{W93,MN98}. Both formalisms
consider terms of the simply-typed $\lambda$-calculus. AFSs use plain
matching and augment the rewrite system with a dedicated rule for
$\beta$-reduction whereas HRSs perform
rewriting modulo $\beta\eta$. In order to keep the matching problem and
therefore rewriting with HRSs manageable, left-hand sides are
typically restricted to Miller's patterns (PATs) \cite{M91}, resulting in
the subclass of pattern rewrite systems (PRSs) \cite{MN98}. Recently, we
introduced another subclass of HRSs by extending PRSs to deterministic
higher-order pattern rewrite systems (DPRSs) \cite{NM25b}. DPRSs
use Yokoyama et al.'s deterministic higher-order patterns (DHPs) as
left-hand sides. Like PATs, DHPs enjoy a deterministic matching problem of
linear complexity \cite{YHT04}.

Since 2010, a category for higher-order rewriting has been part of the
annual international termination competition (TermCOMP).%
\footnote{\url{https://termination-portal.org/wiki/Termination_Competition}}
Except for a demonstration category for HRSs in 2018, the higher-order
category has always considered AFS semantics. In 2024, the
termination problem database (TPDB) was ported to the new ARI format
which is also used in the confluence competition.%
\footnote{\url{https://project-coco.uibk.ac.at/}} A discussion about
the higher-order format ensued, and a proposal by Kop was implemented.%
\footnote{\url{https://github.com/orgs/TermCOMP/discussions/86}} This
year, TermCOMP provides a reference describing the syntax and
semantics of its categories for the first time \cite{termcompref}. In
that document, Kop defines the semantics of the higher-order category
as a restricted version of her algebraic functional systems with
meta-variables (AFSMs) \cite{K12} which we shall refer to as
\emph{simply-typed meta-substitution rewrite systems} (STMRSs). The
definition of STMRSs is given such that other popular flavors
of higher-order rewriting can be accommodated easily. A subcategory
for AFSs is established by simply restricting the eligible benchmarks.
For PRSs, Kop defines a restricted class of STMRSs where
termination of the strategy which prefers $\beta$-reduction
coincides with universal termination of the corresponding PRS.
This idea is based on a result in Kop's PhD thesis \cite[Theorem 3.5]{K12}.

In this paper, we extend STMRSs such that
they can model a more general class of HRSs including our DPRSs. To
that end, we introduce the ad-hoc definition of
\emph{extended pattern rewrite systems} (EPRSs) which constitute a
generalization of DPRSs
where we do not enforce deterministic matching. EPRSs are not meant to
be useful in their own right, but they simplify the development in
this paper and allow us to focus on the properties of DPRSs which are
needed in order to obtain our main goal. Instead of just proving the
relationship between EPRSs and STMRSs with respect to termination,
we show that any step in an EPRS can
be modeled by a step followed by $\beta$-normalization in
STMRSs. This implementation result lays the foundation for comparing also
other properties besides termination. Furthermore, our result extends and
strengthens similar results in the literature such 
as Kop's implementation of PRSs in AFSMs \cite[Section 3.3]{K12}
as well as van Oostrom and van Raamsdonk's implementation of PRSs
in Klop's combinatory systems with explicit $\beta$ \cite{vOvR93}.

\section{Preliminaries}
\label{sec:preliminaries}

In this paper, we consider rewrite systems over
simply-typed $\lambda$-terms \cite{C40,BDS13}.
Given a set of \emph{sorts} $\xS$ ($\iota$, $\kappa$),\footnote{%
This means that in the remainder,
$\iota$ and $\kappa$ are used exclusively to represent arbitrary sorts.
We will continue to define naming conventions in this fashion.}
the set of
simple \emph{types} $\xT$ ($\sigma$, $\tau$) is defined as
$\xS \subseteq \xT$ and $\sigma \to \tau \in \xT$ whenever
$\sigma, \tau \in \xT$. We follow the usual convention that $\to$ is
right-associative, so all types can be denoted in a form
$\sigma_1 \to \cdots \to \sigma_n \to \iota$.
We assume an infinite
set of typed variables $\xV$ ($x$, $y$, $z$, $w$) such that there are
infinitely many variables of every type as well as a signature $\xF$
($f$, $g$) consisting of typed function symbols disjoint from
$\xV$. A \emph{head} ($h$) is either a variable or a function symbol.
We write $h : \sigma \in \xF$ ($h : \sigma \in \xV$) if $h$ is a head
of type $\sigma$. If this information is not
important, we sometimes omit the type and just write $h \in \xF$
($h \in \xV$). The set $\Lambda_\sigma$ of
\emph{simply-typed $\lambda$-terms} ($s$, $t$, $u$, $v$) of type $\sigma$
is defined inductively together with the set $H_\sigma$ of \emph{head-terms}
of type $\sigma$:
\begin{itemize}
\item
$h \in H_\sigma$ if $h : \sigma \in \xF \cup \xV$,
\item
$\lambda x.s \in H_{\tau \to \sigma}$ if $x : \tau \in \xV$ and
$s \in \Lambda_\sigma$, and
\item
$s t_1 \cdots t_n \in \Lambda_\sigma$ if
$s \in H_{\tau_1 \to \cdots \to \tau_n \to \sigma}$ and
$t_i \in \Lambda_{\tau_i}$ for $1 \leq i \leq n$.
\end{itemize}
We obtain $H_\sigma \subseteq \Lambda_\sigma$ by setting $n = 0$ in the
final case above. This presentation of (simply-typed) $\lambda$-terms is due
to Wadsworth; a formal justification is given in \cite[Lemma 8.3.7]{B84}.
If $s \in \bigcup_{\sigma \in \xT}\Lambda_\sigma$ we say that
$s$ is a \emph{term}. We abbreviate
$\lambda x_1. \dots \lambda x_n.s$ to $\lambda x_1,\dots,x_n.s$ and
denote the set of free variables of a term $s$ by $\fv{s}$.
A variable $x$ is \emph{fresh} for a term $s$ if $x \notin \fv{s}$.
The proper subterm relation is defined as follows:
$\lambda x.s \subt u$ if $s \subteq u$ and
$st_1\cdots t_n \subt u$ if $t_i \subteq u$ for
some $1 \leq i \leq n$. Here, ${\subteq} = {\subt} \cup {=}$.
If $s \subteq t$ we say that $t$ is a \emph{subterm} of $s$.
A \emph{context} ($C$) is a term containing exactly one occurrence of
the special symbol $\square$ which can assume any type. We write
$C[s]$ to denote the term where $\square$ has been replaced by the
term $s$ of the same type.

Finite mappings from variables to terms of the same type
are called \emph{substitutions} ($\gamma$). Given a substitution
$\gamma = \SET{x_1 \mapsto s_1,\,\dots,\,x_n \mapsto s_n}$
we define
$\dom{\gamma} = \SET{\seq{x}}$,
$\im{\gamma} = \SET{\seq{s}}$ and
$\fv{\gamma} = \bigcup_{i=1}^n \fv{s_i}$. We say that a variable $x$
is \emph{fresh} for a substitution $\gamma$ if
$x \notin \dom{\gamma} \cup \fv{\gamma}$. By $\gamma|_V$ we denote
the restriction of $\gamma$ such that
$\dom{\gamma|_V} = \dom{\gamma} \cap V$. Given a term $s$ and a
substitution $\gamma$, its application $s\gamma$ binds stronger than
the abstraction operator and is defined by induction
on $s$: $h\gamma = \gamma(h)$ if $h \in \dom{\gamma}$, $h\gamma = h$
if $h \notin \dom{\gamma}$,
$(\lambda x.s)\gamma = \lambda z.s\SET{x \mapsto z}\gamma$ where $z$
is fresh for $\gamma$ and $s$ (capture-avoidance), and
$(s t_1 \cdots t_n)\gamma = (s\gamma) (t_1\gamma) \cdots (t_n\gamma)$.

Every term $s$ has a unique $\beta$-normal form $s\bnf$. We will
often use the fact that $\beta$-reduction is confluent in proofs.
We write $s \in \nfb$ if $s$ is already in $\beta$-normal form. The unique
$\eta$-long form $s\eex$ of a term $s$ is
obtained by exhaustively applying restricted $\eta$-expansion \cite{A93}.
We write $s \in \nfel$ if $s$ is in $\eta$-long normal from.
If $s \in \nfb$ and $s \in \nfel$ we say that $s$ is in
its unique \emph{$\beta\eta$-long normal form} and write $s \in \nfbel$.
We follow the convention that $\bnf$ and $\eex$ 
bind stronger than the abstraction operator.
We write $s\lnf$ for $s\bnf\eex$.
It is convenient to work with terms in $\nfel$ because $\nfel$
is closed under substitution and $\beta$-reduction, see e.g.\
\cite[Proposition 4.2.9]{vR96}. As a consequence, we immediately
obtain the following result.

\begin{lemma}
\label{lem:nfbelsubst}
If $s \in \nfbel$ and $\im{\gamma} \subseteq \nfbel$ then
$s\gamma\bnf \in \nfbel$.
\end{lemma}

Finally, we recall Mayr and Nipkow's definition of HRSs \cite{MN98} using
our notation. A \emph{rule} $\ell \R r$ satisfies the following conditions:
$\ell, r : \iota$ with $\iota \in \xS$,
$\ell = f s_1 \cdots s_n, r \in \nf{\beta}$ with $f \in \xF$ and
$\fv{r} \subseteq \fv{\ell}$. A \emph{higher-order rewrite system}
(HRS) is a set of rules. Given an HRS $\xR$, its rewrite relation
${\RbR} \subseteq \nfbel \times \nfbel$ is defined as follows:
$s \RbR t$ if there exist $\ell \R r \in \xR$, a context $C$
and a substitution $\gamma$ with $\im{\gamma} \subseteq \nfbel$ such that
$s = C[\ell\gamma\lnf]$ and $t = C[r\gamma\lnf]$.
In contrast to the definition in \cite{MN98}, we do not demand
that right-hand sides of rules are in $\nfbel$ in order to simplify
the relationship with STMRSs later. Note that this does not have
any effect on the rewrite relation as restricted $\eta$-expansion
is applied anyway.

\section{Extended Pattern Rewrite Systems}

In this section, we present a subclass of HRSs which
generalizes Nipkow's PRSs \cite{MN98}. The rewrite systems we
want to target are the DPRSs recently introduced by us \cite{NM25b}.
DPRSs use Yokoyama et al.'s DHPs \cite{YHT04}
as left-hand sides. While DHPs and therefore DPRSs
have desirable properties with respect to matching,
unification and ultimately confluence, the emphasis
on termination in this paper allows us to ignore
a large part of the involved definition of DHPs.
Thus, the following definition subsumes both PATs and DHPs but is not
meant to be practically useful.

\begin{definition}
\label{def:epat}
A term $s \in \nf{\beta}$ is an \emph{extended pattern} \textup{(EPAT)} if
every free occurrence of a variable $F$ is in a subterm
$F u_1 \cdots u_n$ such that each $u_i$ is $\eta$-equivalent to
a non-abstraction and does not contain free variables.
\end{definition}

\begin{example}
Consider the sort $\m{a}$, the function symbol $\m{f} : \m{a} \to \m{a}$
as well as the variables
$x : \m{a}$, $y : \m{a}$ and
$Z : \m{a} \to (\m{a} \to \m{a}) \to \m{a} \in \xV$.
The terms $\lambda x.Z\:x\:\m{f}$,
$\lambda x.Z\:(\m{f}\:x)\:\m{f}$ and
$\lambda x.Z\:(\m{f}\:x)\:(\lambda w.\m{f}\:w)$ are EPATs while
$\lambda x.Z\:y\:\m{f}$ and
$\lambda x.Z\:x\:(\lambda w.\m{f}\:(\m{f}\:w))$ are not.
\end{example}

\begin{definition}
A rule $\ell \R r$ is an \emph{extended pattern rule}
if $\ell$ is an \textup{EPAT}. An \emph{extended pattern rewrite system}
\textup{(EPRS)} is a set of extended pattern rewrite rules.
\end{definition}

\begin{example}
\label{exa:eprs}
Consider the sorts $\m{a}$, $\m{b}$ and $\m{u}$, the function symbols
$\m{inl} : \m{a} \to \m{u}$, $\m{inr} : \m{b} \to \m{u}$ and
$\m{case} : \m{u} \to (\m{a} \to \m{u}) \to (\m{b} \to \m{u}) \to \m{u}$
as well as the variables $x : \m{a}$, $y : \m{b}$, $z : \m{u}$,
$F : \m{a} \to \m{u}$,
$G : \m{b} \to \m{u}$ and $H : \m{u} \to \m{u}$. The EPRS $\xR$ consisting
of the rules
\begin{align*}
\m{case}\:(\m{inl}\:x)\:F\:G &\R F\:x \\
\m{case}\:(\m{inr}\:y)\:F\:G &\R G\:y \\
\m{case}\:z\:(\lambda w.H\:(\m{inl}\:w))\:(\lambda w.H\:(\m{inr}\:w))
&\R H\:z
\end{align*}
models a case distinction operation on the disjoint union datatype
$\m{u} = \m{a} \uplus \m{b}$. For simplicity, we only provide one definition
of $\m{case}$ which returns terms of type $\m{u}$. This system is given as
an interesting example of a non-PRS in Nipkow's
paper introducing PRSs \cite{N91}. In particular,
matching with $\xR$ is still deterministic as it
falls into the class of DPRSs \cite{NM25b}.
Termination of $\xR$ has been established
by van de Pol \cite[p.~74]{vdP96}.
Note that the only benchmark in TPDB corresponding to $\xR$,
\href{https://github.com/TermCOMP/TPDB-ARI/blob/master/Higher_Order_%
Rewriting/Mixed_HO_10/sdu.ari}%
{\texttt{Higher\_Order\_Rewriting/Mixed\_HO\_10/sdu.ari}},
is not an adequate model of
$\xR$ as it uses an undefined function symbol in order to
express the third rule in a PRS.
\end{example}

\section{Simply-Typed Meta-Substitution Rewrite Systems}

We now present an extension of the current semantics of the higher-order
category in TermCOMP \cite{termcompref}.
The grammar for higher-order benchmarks in TPDB
\begin{align*}
\text{hotrs} &::= \texttt{(format\;higher-order)}\;\text{sort}{+}\;
\text{fun}{+}\;\text{rule}{+} \\
\text{sort} &::= \texttt{(sort}\;\text{identifier}\texttt{)} \\
\text{fun} &::= \texttt{(fun}\;\text{identifier}\;\text{type}\texttt{)} \\
\text{type} &::= \text{identifier} \,\mid\,
\texttt{(->}\;\text{type}{+}\;\text{identifier}\texttt{)} \\
\text{term} &::= \text{identifier} \,\mid\,
\texttt{(}\text{identifier}\;\text{term}{+}\texttt{)} \,\mid\,
\texttt{(lambda}\;\texttt{(}\text{var}{+}\texttt{)}\;\text{term}\texttt{)}
\\
\text{var} &::= \texttt{(}\text{identifier}\;\text{type}\texttt{)} \\
\text{rule} &::= \texttt{(rule}\:\text{term}\;\text{term}\texttt{)}
\end{align*}
is an extension of the grammar for many-sorted term rewrite systems
in the ARI format.\footnote{\url{https://project-coco.uibk.ac.at/ARI}}
The grammar is compatible with our presentation of simply-typed
$\lambda$-terms and forces terms in rules to be in $\nfb$. The following
definition generalizes the restrictions on rules for the higher-order
category in TermCOMP by allowing EPATs instead of just PATs as left-hand
sides in rules.

\begin{definition}
An \emph{\textup{STMRS} rule} $\ell \RR r$ satisfies the following
conditions:
\begin{itemize}
\item
$\ell$ and $r$ have the same type,
\item
$\ell, r \in \nf{\beta}$,
\item
$\fv{r} \subseteq \fv{\ell}$,
\item
$\ell = f s_1 \cdots s_n$ with $f \in \xF$,
\item
$\ell$ is an \textup{EPAT}, and
\item
variables are used \emph{consistently}, i.e., there is a function
$\ar : \fv{\ell} \to \mathbb{N}$ such that
\begin{itemize}
\item
if $\ell \subteq x s_1 \cdots s_n$ and $x \in \fv{\ell}$ then $\ar(x) = n$,
and
\item
if $r \subteq x s_1 \cdots s_n$ and $x \in \fv{\ell}$ then $\ar(x) \leq n$.
\end{itemize}
\end{itemize}
\end{definition}

Since the arity of symbols is not fixed in the syntax employed in TPDB,
we have to ensure consistent usage of variables to facilitate the
definition of meta-substitutions.

\begin{definition}
Given a rule $\ell \RR r$ and its induced arity function $\ar$, a
\emph{meta-substitution} $\theta$ maps each variable
$x \in \fv{\ell}$ to $\langle \seq[\ar(x)]{y} \rangle u$ where
$u : \tau$ whenever $x : \sigma_1 \to \cdots \to \sigma_{\ar(x)} \to \tau$.
Intuitively, $\langle \seq[\ar(x)]{y} \rangle$ denotes
meta-$\lambda$-abstraction which is eliminated in meta-substitution
application. Hence, we identify $\langle \rangle u$ with $u$. Furthermore,
we define $\dom{\theta} = \fv{\ell}$,
$\im{\theta} = \SET{\lambda \seq[\ar(x)]{y}.u \mid
\text{$x \in \dom{\theta}$ and
$\theta(x) = \langle \seq[\ar(x)]{y} \rangle u$}}$
and $\fv{\theta} = \bigcup \SET{\fv{s} \mid s \in \im{\theta}}$.
As for regular substitutions, a variable $x$ is \emph{fresh} for a
meta-substitution $\theta$ if $x \notin \dom{\theta} \cup \fv{\theta}$.
Given $s \in \SET{\ell, r}$ and a meta-substitution $\theta$, its
application $s\theta$ binds stronger than the abstraction operator and is
defined inductively on terms in $\nfb$ as follows:
\begin{itemize}
\item
$(h s_1 \cdots s_n)\theta = (u\gamma) \cdot (s_{\ar(h)+1}\theta) \cdots
(s_n\theta)$ if $\theta(h) = \langle \seq[\ar(h)]{y} \rangle u$
\item[] where
$\gamma = \SET{y_1 \mapsto s_1\theta,\,\dots,\,y_{\ar(h)} \mapsto
s_{\ar(h)}\theta}$ is an ordinary substitution,
\item
$(h s_1 \cdots s_n)\theta = h (s_1\theta) \cdots (s_n\theta)$ if
$h \notin \dom{\theta}$, and
\item
$(\lambda x.s)\theta = \lambda z.s\SET{x \mapsto z}\theta$ where $z$ is
fresh for $\theta$ and $s$.
\end{itemize}
\end{definition}

Since meta-substitutions depend on arities of free variables,
they are always scoped within a specific rule which provides
the corresponding arity function. For globally fixed arities,
the notion of \emph{substitutes} as used in combinatory reduction
systems corresponds to meta-substitutions \cite{K93,KvOvR93}.

\begin{example}
\label{exa:subst}
Let $\m{a} \in \xS$,
$x : \m{a} \to \m{a}$, $y : \m{a}$,
$Z : (\m{a} \to \m{a}) \to \m{a} \in \xV$,
$\m{0} : \m{a}$, $\m{1} : \m{a}$, $\m{g} : \m{a} \to \m{a} \to \m{a}$,
$\m{f} : (\m{a} \to \m{a}) \to \m{a} \to \m{a} \in \xF$
and consider the rule
$\m{f}\:(\lambda y.Z\:(\m{g}\:y))\:\m{0} \RR Z\:(\lambda y.y)$
from which we infer $\ar(Z) = 1$. For the meta-substitution
$\theta = \SET{Z \mapsto \langle x \rangle (x\:\m{1})}$ we have
$(\m{f}\:(\lambda y.Z\:(\m{g}\:y))\:\m{0})\theta =
\m{f}\:(\lambda y.\m{g}\:y\:\m{1})\:\m{0}$ and
$(Z\:(\lambda y.y))\theta = (\lambda y.y)\:\m{1}$.
\end{example}

For an STMRS rule $\ell \RR r$ and a meta-substitution $\theta$ with
$\im{\theta} \in \nfb$ we have $\ell\theta \in \nfb$ by construction;
a related statement is proved later in \lemref{lhs}.
As we have seen in the previous example, this need not
be the case for $r\theta$. In particular, meta-substitution
application corresponds to ordinary substitutions with
a subsequent \emph{development} step \cite{B84,vOvR93}
as opposed to rewriting to $\beta$-normal form.
In EPATs, substituting $\beta$-normal terms for free variables
with non-abstractions as arguments cannot create $\beta$-redexes,
so a development step is enough for $\beta$-normalization.

\begin{definition}
A \emph{simply-typed meta-substitution rewrite system}
\textup{(STMRS)} is a set of \textup{STMRS} rules. Given a \textup{STMRS}
$\xR$, its rewrite relation
$\RRb{\xR \cUp \beta}$ is defined as follows for arbitrary contexts $C$:
\begin{itemize}
\item
$C[\ell\theta] \RRb{\xR} C[r\theta]$ if $\ell \RR r \in \xR$ and
$\theta$ is a corresponding meta-substitution, and
\item
$C[(\lambda x.s) t_1 \cdots t_n] \RRb{\beta}
C[s\SET{x \mapsto t_1} t_2 \cdots t_n]$ for $n \geq 1$.
\end{itemize}
We use $\RRab{\beta}{*}$ for repeated applications of $\beta$-reduction and
$\RRab{\beta}{!}$ for rewriting to $\beta$-normal form.
\end{definition}

AFSs can be expressed by STMRSs by demanding that free variables
always have arity $0$. In this case, meta-substitution corresponds
to ordinary substitution. Although free variables may have
a strictly positive arity in AFSs, this is usually not what is
intended: The left-hand side
$\m{d}\:(\lambda z.\m{sin}\:(F\:z))$ from the well-known
HRS\footnote{see e.g.~\href{https://github.com/TermCOMP/TPDB-ARI/blob/%
master/Higher_Order_Rewriting/Hamana_Kikuchi_18/h16.ari}%
{\texttt{Higher\_Order\_Rewriting/Hamana\_Kikuchi\_18/h16.ari}} in TPDB}
modeling parts of differential calculus
should match $\m{d}\:(\lambda z.\m{sin}\:z)$, but this is not
possible with plain matching as used in AFSs.

For EPRSs, the relation to STMRSs is more intricate. Clearly,
the EPRS presented in \exaref{eprs} is also an STMRS. In general,
any STMRS which only consists of rules typed by a sort
is also an EPRS. For EPATS with non-abstractions
as arguments of free variables, we have already
seen that matching with $\beta$-normal meta-substitutions
(as done in STMRSs) is as powerful as full higher-order matching (as
done in EPRSs) because the implicit $\beta$-reduction in
meta-substitution application is sufficient for $\beta$-normalization.
However, applying meta-substitutions to right-hand sides does not
yield $\beta$-normal terms in general. Hence, applying a rewrite step in an
STMRS and its corresponding EPRS might yield different terms. On top
of that, rules are used in their $\eta$-expanded form in EPRS rewrite
steps while no such transformation is done for STMRS rewrite steps.
In the next section, we will establish a syntactic subclass of STMRSs (and
therefore an syntactically definable subset of higher-order benchmarks
in TPDB) which can implement their corresponding EPRSs by normalized
rewriting \cite{M96} with respect to $\beta$-reduction.

\section{Implementing EPRSs in STMRSs}

The following two definitions describe the class of STMRSs which are
able to implement their corresponding EPRSs. As such, these definitions
extend the restriction proposed for a PRS subcategory of the higher-order
category in TermCOMP \cite{termcompref}.

\begin{definition}
Given a set of variables $W$ and an arity function $\ar$,
\emph{(quasi-)well-behaved} terms with respect to $W$ and $\ar$ are
defined inductively as follows:
\begin{itemize}
\item
$\lambda x.s$ is (quasi-)well-behaved if so is $s$,
\item
$h s_1 \cdots s_n$ is (quasi-)well-behaved if its type is a sort,
$h \in \xF \cup (\xV \setminus W)$ and all $s_i$ are (quasi-)well-behaved,
\item
$x s_1 \cdots s_n$ with $x \in W$ is quasi-well-behaved if all $s_i$ are
quasi-well-behaved,
\item
$x s_1 \cdots s_n$ with $x \in W$ is (quasi-)well-behaved if 
\begin{itemize}
\item
for all $1 \leq i \leq \ar(x)$, $s_i = h_i t_1 \cdots t_{m_i}$
with all $t_j$ (quasi-)well-behaved, and
\item
for all $\ar(x) < i \leq n$, $s_i$ is (quasi-)well-behaved.
\end{itemize}
\end{itemize}
\end{definition}

\begin{example}
Consider the sort $\m{a}$, the function symbols
$\m{f} : (\m{a} \to \m{a}) \to \m{a} \to \m{a}$,
$\m{g} : \m{a} \to \m{a}$ as well as the variables $x : \m{a}$,
$F : \m{a} \to \m{a}$ and $G : (\m{a} \to \m{a}) \to \m{a} \to \m{a}$.
Furthermore, let $W = \SET{F,G}$ with $\ar(F) = 0$ and $\ar(G) = 1$.
The terms
\begin{gather*}
\m{f}\:F\:x \qquad \m{f}\:(\lambda z.F\:z)\:x \qquad
\m{f}\:(G\:\m{g})\:x \qquad
\m{f}\:(\lambda z.G\:\m{g}\:z)\:x
\end{gather*}
are (quasi-)well-behaved with respect to $W$ an $\ar$,
\begin{gather*}
\m{f}\:(G\:(\lambda w.\m{g}\:w))\:x \qquad
\m{f}\:(\lambda z.G\:(\lambda w.\m{g}\:w)\:z)\:x
\end{gather*}
are quasi-well-behaved but not well-behaved
with respect to $W$ and $\ar$ and the term
$\m{f}\: F$ is
neither well-behaved nor quasi-well-behaved with respect to $W$ and $\ar$.
\end{example}

Intuitively, (quasi-)well-behaved terms with respect to $W$ and $\ar$
are in $\beta\eta$-long normal form up to variables $x \in W$ which do
not have to be fully applied. This allows for succinct presentations
of (quasi-)well-behaved terms. In addition, quasi-well-behaved terms
allow the first $\ar(x)$ arguments of variables $x \in W$ to be
fully $\eta$-reduced at the root. For well-behaved terms,
this is enforced, so the first $\ar(x)$ arguments
of variables $x \in W$ may not be abstractions.
While it is sufficient for right-hand sides of rules to be
quasi-well-behaved, we will demand left-hand sides of rules to be
well-behaved; in combination with the restriction to EPATs,
this ensures that the application of meta-substitutions $\theta$ with
$\dom{\theta} = W$ does not create $\beta$-redexes.

\begin{definition}
A rule $\ell \RR r$ with arity function $\ar$ is well-behaved if
it is typed by a sort,
\begin{itemize}
\item
$\ell$ is well-behaved with respect to $\fv{\ell}$ and $\ar$, and
\item
$r$ is quasi-well-behaved with respect to $\fv{\ell}$ and $\ar$.
\end{itemize}
An \textup{STMRS} $\xR$ is well-behaved if all its rules are well-behaved.
\end{definition}

Since all rules in well-behaved STMRS can be typed by sorts, a well-behaved
STMRS can also be viewed as an EPRS. The additional constraints on
(quasi-)well-behaved terms ensure a correspondence between the rewrite
relations of well-behaved STMRS and their corresponding EPRSs. The upcoming
definition connects meta-substitutions and ordinary substitutions.

\begin{definition}
Given a meta-substitution $\theta$, let $\theta^\lambda$ be the
substitution on the same domain where
$\theta^\lambda(x) = \lambda \seq{x}.u$ whenever
$\theta(x) = \langle \seq{x} \rangle u$.
Furthermore, given a substitution $\gamma$ with
$\im{\gamma} \subseteq \nfbel$ and a rule $\ell \RR r$ with induced arity
function $\ar$ we define the meta-substitution $\gamma^{-\lambda}$
as follows:
\[
\gamma^{-\lambda}(x) = \begin{cases}
\langle \seq[\ar(x)]{x} \rangle \lambda x_{\ar(x)+1}, \dots, x_n.u
&\text{if $\gamma(x) = \lambda \seq{x}.u$} \\
\langle \seq[\ar(x)]{x} \rangle \lambda x_{\ar(x)+1}, \dots, x_n.x \cdot
x_1\eex \cdots x_n\eex
&\text{if $x \in \fv{\ell} \setminus \dom{\theta}$}
\end{cases}
\]
Here, we assume $x\eex = \lambda \seq{x}.x \cdot x_1\eex \cdots x_n\eex$.
\end{definition}

\begin{example}
\label{exa:stmrsVSeprs}
Recall the rule $\ell \RR r$ and the meta-substitution $\theta$
from \exaref{subst} and let $\xR$ be the well-behaved STMRS
$\SET{\ell \RR r}$. As we have already seen, $\xR$ can also be viewed as
an EPRS. By definition, $\theta^\lambda = \SET{z \mapsto \lambda x.x\:1}$.
The following comparison of rewriting with $\RRb{\xR}$ and $\Rb{\xR}$
suggests that EPRSs can be implemented in STMRSs.
\begin{align*}
\ell\theta &= \m{f}\:(\lambda y.\m{g}\:y\:\m{1})\:\m{0}
\RRb{\xR} (\lambda y.y)\:\m{1} = r\theta \RRab{\beta}{!} \m{1} \\
\ell\theta^\lambda\lnf &= \m{f}\:(\lambda y.\m{g}\:y\:\m{1})\:\m{0}
\Rb{\xR} \m{1} = r\theta^\lambda\lnf
\end{align*}
\end{example}

The upcoming theorem, our main result, states the implementation result
evidenced by the previous example.

\begin{restatable}{theorem}{main}
\label{thm:main}
Let $\xR$ be a well-behaved \textup{STMRS} and $s \in \nfbel$.
We have $s \RRb{\xR} \cdot \RRab{\beta}{!} t$ if and only if $s \Rb{\xR} t$
when viewing $\xR$ as an \textup{EPRS}.
\end{restatable}

\section{Proof of \thmref{main}}

We start with three lemmata about left-hand sides of well-behaved
STMRSs with respect to terms in $\nfbel$. We will use
these lemmata to show that well-behaved STMRSs and EPRSs
can match the same terms in $\nfbel$.

\begin{lemma}
\label{lem:lhs}
Let $\ell \RR r$ be a well-behaved \textup{STMRS} rule and $\theta$ a
corresponding meta-substitution. If $\im{\theta} \subseteq \nfbel$ then
$\ell\theta \in \nfbel$.
\end{lemma}

\begin{proof}
Let $\ar$ be the arity function induced by $\ell \RR r$. We prove
the result by induction on subterms $s$ of $\ell$ with respect to the
definition of well-behaved terms.
\begin{itemize}
\item
Let $s = \lambda x.v$. Without loss of generality, $x$ is fresh for
$\theta$, so $s\theta = \lambda x.v\theta$ and the result follows from the
induction hypothesis together with the fact that $\nfbel$ is closed under
the abstraction operator.
\item
Let $s = h s_1 \cdots s_n$ with $h \in \xF \cup (\xV \setminus \fv{\ell})$
and its type be a sort. The
induction hypothesis yields that $s_i\theta \in \nfbel$ for all
$1 \leq i \leq n$. By definition, $h \notin \dom{\theta}$, so
$s\theta = h (s_1\theta) \cdots (s_n\theta)$ and the result
follows immediately from the induction hypothesis.
\item
Let $s = x s_1 \cdots s_n$ where $x \in \fv{\ell}$.
Since $s$ is a subterm of the left-hand side of an STMRS rule,
$n = \ar(x)$. By assumption, $x \in \dom{\theta}$, so
$\theta(x) = \langle \seq{y} \rangle u$
with $u \in \nfbel$. Since $\ell$ is an EPAT, we
know that $\fv{s_i} \cap \fv{\ell} = \varnothing$ and therefore
$s_i = s_i\theta$ for all $1 \leq i \leq n$. Let
$\gamma = \SET{y_1 \mapsto s_1,\,\dots,\,y_n \mapsto s_n}$.
It is sufficient to prove that $u\gamma \in \nfbel$.
We proceed by an inner induction
on $u$. Let $u = \lambda x.v$. Without
loss of generality, $x$ is fresh for $\gamma$. The inner
induction hypothesis yields that $v\gamma \in \nfbel$, so
$u\gamma = \lambda x.v\gamma \in \nfbel$ as $\nfbel$ is closed
under the abstraction operator. Now let
$u = h v_1 \cdots v_k$. The inner induction hypothesis
yields that $v_j\gamma \in \nfbel$ for all $1 \leq j \leq k$. If
$h \notin \SET{\seq{y}}$ then the result follows
immediately. Otherwise, $h = y_i$ for some $i$. By assumption,
$s_i = h_i t_1 \cdots t_m$, and the outer induction
hypothesis ensures that $t_j = t_j\theta \in \nfbel$ for all
$1 \leq j \leq m$. Hence,
$u\gamma = h_i t_1 \cdots t_m (v_1\gamma) \cdots (v_k\gamma) \in
\nfbel$ as desired. \qedhere
\end{itemize}
\end{proof}

\begin{lemma}
\label{lem:metasubstbelnf}
Let $\ell \RR r$ be a well-behaved \textup{STMRS} rule and $\theta$ be
a corresponding meta-substitution. If $\ell\theta \in \nfbel$ then
$\im{\theta} \subseteq \nfbel$.
\end{lemma}

\begin{proof}
Let $\ar$ be the arity function induced by $\ell \RR r$.
For $\ell \subteq s$ we prove that $s\theta \in \nfbel$ implies
$\im{\theta|_{\fv{s}}} \subseteq \nfbel$ by
induction on the definition of well-behaved terms.
The result then follows by taking $s = \ell$.
\begin{itemize}
\item
Let $s = \lambda x.t$. Without loss of generality, $x$ is fresh for
$\theta$, so $s\theta = \lambda x.t\theta$. Since
$\dom{\theta|_{\fv{t}}} = \dom{\theta|_{\fv{s}}}$, the result immediately
follows from the induction hypothesis.
\item
Let $s = h s_1 \cdots s_n$ with $h \in \xF \cup (\xV \setminus \fv{\ell})$
and its type by a sort. By definition, $h \notin \dom{\theta}$, so
$s\theta = h \cdot s_1\theta \cdots s_n\theta$. The induction hypothesis
yields
$\im{\theta|_{\fv{s_i}}} \subseteq \nfbel$ for all $1 \leq i \leq n$.
Together with $\fv{s} = \bigcup_{i=1}^n \fv{s_i}$ we obtain
$\im{\theta|_{\fv{s}}} \subseteq \nfbel$.
\item
Let $s = x s_1 \cdots s_n$ where $x \in \fv{\ell}$.
Since $s$ is the subterm of a left-hand side of an STMRS rule,
$n = \ar(x)$. By assumption, $x \in \dom{\theta}$, so
$\theta(x) = \langle \seq{y} \rangle u$. Since $\ell$ is an EPAT, we know
that $\fv{s_i} \cap \dom{\theta} = \varnothing$ and therefore
$s_i = s_i\theta$ for all $1 \leq i \leq \ar(x)$. Hence,
$s\theta = u\SET{y_1 \mapsto s_1,\,\dots,\,y_n \mapsto s_n}$. Since
$s\theta \in \nfbel$ and $s_i$ is not an abstraction for all
$1 \leq i \leq \ar(x)$, also $u \in \nfbel$. \qedhere
\end{itemize}
\end{proof}

\begin{lemma}
\label{lem:minuslam}
Let $\ell \RR r$ be a well-behaved \textup{STMRS} rule.
If $\im{\gamma} \subseteq \nfbel$ then
$\ell(\gamma^{-\lambda})^\lambda\lnf = \ell\gamma\lnf$
\end{lemma}

\begin{proof}
By definition, $\dom{\gamma} \subseteq \dom{\gamma^{-\lambda}}$. Let $\ar$
be the arity function induced by $\ell \RR r$. We prove the result by
induction on subterms $s$ of $\ell$ with respect to the definition
of well-behaved terms.
\begin{itemize}
\item
Let $s = \lambda x.v$. Without loss of generality, $x$ is fresh for
$(\gamma^{-\lambda})^\lambda$, so $s(\gamma^{-\lambda})^\lambda\lnf =
\lambda x.v(\gamma^{-\lambda})^\lambda\lnf =
\lambda x.v\gamma\lnf = s\gamma\lnf$ follows from the induction
hypothesis.
\item
Let $s = h s_1 \cdots s_n$ with $h \in \xF \cup (\xV \setminus \fv{\ell})$
and its type be a sort. By definition, $h \notin \dom{\theta}$, so
\begin{align*}
s(\gamma^{-\lambda})^\lambda\lnf
&= h (s_1(\gamma^{-\lambda})^\lambda\lnf) \cdots
(s_n(\gamma^{-\lambda})^\lambda\lnf) \\
&= h (s_1\gamma\lnf) \cdots (s_n\gamma\lnf) \tag{IH} \\
&= s\gamma\lnf
\end{align*}
as desired.
\item
Let $s = x s_1 \cdots s_n$ where $x \in \fv{\ell}$ and
$x\eex = \lambda \seq[m]{y}.x (y_1\eex) \cdots (y_m\eex)$ with $m \geq n$.
If $x \in \dom{\gamma}$ then $\gamma(x) = (\gamma^{-\lambda})^\lambda(x) =
\lambda \seq[m]{y}.u \in \nfbel$ by assumption. By definition EPATs,
$s_i(\gamma^{\lambda})^\lambda = s_i\gamma = s_i$ for all
$1 \leq i \leq n$. Hence,
$s(\gamma^{-\lambda})^\lambda = s\gamma$ and therefore
$s(\gamma^{-\lambda})^\lambda\lnf = s\gamma\lnf$ as desired.
Now consider $x \notin \dom{\gamma}$. By definition
$(\gamma^{-\lambda})^\lambda(x) = x\eex$. Hence,
\[
s(\gamma^{-\lambda})^\lambda\lnf
= (x\eex s_1 \cdots s_n)\lnf
= (x s_1 \cdots s_n)\lnf
= s\gamma\lnf
\]
as desired. \qedhere
\end{itemize}
\end{proof}

Next, we show the connection between the application
of a meta-substitution and its corresponding ordinary
substitution on a (quasi-)well-behaved term. This is
needed in order to synchronize the resulting terms
after performing the same step with a well-behaved
STMRS and its corresponding EPRS.

\begin{lemma}
\label{lem:both}
Let $\ell \RR r$ be a well-behaved \textup{STMRS} rule and $\theta$ a
corresponding meta-substitution. If $\im{\theta} \subseteq \nfbel$ then
$t\theta^\lambda\bnf = t\theta\bnf \in \nfbel$ for $t \in \SET{\ell, r}$.
\end{lemma}

\begin{proof}
Let $\ar$ be the arity function induced by $\ell \RR r$. We prove
the result by induction on subterms $s$ of $t \in \SET{\ell, r}$
with respect to the definition of (quasi-)well-behaved terms.
\begin{itemize}
\item
Let $s = \lambda x.v$. Without loss of generality, $x$ is fresh for
$\theta$. By the induction hypothesis,
$v\theta^\lambda\bnf = v\theta\bnf \in \nfbel$ and therefore
$s\theta^\lambda\bnf = \lambda x.v\theta^\lambda\bnf =
\lambda x.v\theta\bnf = s\theta\bnf \in \nfbel$ as
$\nfbel$ is closed under the abstraction operator.
\item
Let $s = h s_1 \cdots s_n$ and its type be a sort. The induction
hypothesis yields $s_i\theta^\lambda\bnf = s_i\theta\bnf \in \nfbel$
for all $1 \leq i \leq n$. By assumption $h \notin \dom{\theta}$, so
$s\theta^\lambda\bnf =
h (s_1\theta^\lambda\bnf) \cdots (s_n\theta^\lambda\bnf) =
h (s_1\theta\bnf) \cdots (s_n\theta\bnf) = s\theta\bnf \in \nfbel$.
\item
Let $s = x s_1 \cdots s_n$ where $x \in \fv{\ell}$. By assumption
$x \in \dom{\theta}$ and $\im{\theta} \subseteq \nfbel$, so
\[
\theta(x) =
\langle \seq[\ar(x)]{y} \rangle \lambda y_{\ar(x)+1}, \dots, y_n.u
\]
with $u \in \nfbel$.
For $\gamma_1 = \SET{y_1 \mapsto s_1\theta^\lambda\bnf,\,\dots,\,
y_n \mapsto s_n\theta^\lambda\bnf}$ and $\gamma_2 =
\SET{y_1 \mapsto s_1\theta\bnf,\,\dots,\,y_n \mapsto s_n\theta\bnf}$
we have $s\theta^\lambda \RRab{\beta}{*} u\gamma_1$ and
$s\theta \RRab{\beta}{*} u\gamma_2$.
Hence, we are left to show
$u\gamma_1\bnf = u\gamma_2\bnf \in \nfbel$.
If all $s_i$ are well-behaved, we immediately
obtain $\gamma_1 = \gamma_2$ with $\im{\gamma_1} \subseteq \nfbel$
from the induction hypothesis. We can now use \lemref{nfbelsubst}
to obtain $u\gamma_1\bnf = u\gamma_2\bnf \in \nfbel$ as desired.
Otherwise, we prove $u\gamma_1\bnf = u\gamma_2\bnf \in \nfbel$
by an inner induction on $u$.
Let $u = \lambda x.v$. Without loss of generality, $x$ is fresh for
$\gamma_1$ and $\gamma_2$, so $u\gamma_1\bnf = \lambda x.v\gamma_1\bnf =
\lambda x.v\gamma_2\bnf = u\gamma_2\bnf \in \nfbel$ follows
from the inner induction hypothesis. Now let $u = h v_1 \cdots v_k$.
The inner induction hypothesis yields
$v_i\gamma_1\bnf = v_i\gamma_2\bnf \in \nfbel$
for all $1 \leq i \leq k$. If $h \notin \SET{\seq{y}}$ then the result
follows immediately. Otherwise, $h = y_i$ for some $i$.
If $i \geq \ar(x)$ then
\begin{align*}
u\gamma_1\bnf
&= (s_i\theta^\lambda\bnf)(v_1\gamma_1\bnf) \cdots (v_k\gamma_1\bnf) \\
&= \bigl((s_i\theta\bnf)(v_1\gamma_2\bnf) \cdots
(v_k\gamma_2\bnf)\bigr)\bnf \in \nfbel \tag{outer + inner IH} \\
&= u\gamma_2\bnf \in \nfbel
\end{align*}
where we use the fact $\nfel$ is closed under $\beta$-reduction.
Otherwise, by assumption,
$s_i = h_i \cdot t_1 \cdots t_m$ and $h_i \notin \fv{\ell} = \dom{\theta}$.
The outer induction hypothesis yields that
$t_j\theta^\lambda\bnf = t_j\theta\bnf \in \nfbel$ for all
$1 \leq j \leq k$. Hence
\begin{align*}
u\gamma_1\bnf
&= h_i (t_1\theta^\lambda\bnf) \cdots (t_m\theta^\lambda\bnf)
(v_1\gamma_1\bnf) \cdots (v_k\gamma_1\bnf) \\
&= h_i (t_1\theta\bnf) \cdots (t_m\theta\bnf) (v_1\gamma_2\bnf) \cdots
(v_k\gamma_2\bnf) \in \nfbel \tag{outer + inner IH} \\
&= u\gamma_2\bnf \in \nfbel
\end{align*}
as desired. \qedhere
\end{itemize}
\end{proof}

Finally, we are ready to prove our main result.

\main*

\begin{proof}
We prove the implications separately.
\begin{itemize}
\item[$\Rightarrow$]
Let $\ell \RR r \in \xR$ and $\ar$ its induced arity function
such that $s = C[\ell\theta]$ for some meta-substitution $\theta$.
\lemref{metasubstbelnf} yields that $\im{\theta} \subseteq \nfbel$.
Thus, we can use Lemmata \ref{lem:both} and \ref{lem:lhs} to obtain
$\ell\theta^\lambda\lnf = \ell\theta\bnf = \ell\theta$
which means we also have $s = C[\ell\theta^\lambda\lnf]$.
By definition, $s = C[\ell\theta] \RRb{\xR} C[r\theta]$ and
$s = C[\ell\theta^\lambda\lnf] \RbR C[r\theta^\lambda\lnf]$.
Since the type of $\ell$ and $r$ is a sort and $s \in \nfbel$,
$C[u] \in \nfbel$ whenever $u \in \nfbel$. From \lemref{both} we obtain
$r\theta^\lambda\lnf = r\theta\bnf$, so we have
$C[r\theta] \RRab{\beta}{!} C[r\theta^\lambda\lnf]$ as desired.
\smallskip
\item[$\Leftarrow$]
We now view the well-behaved STMRS $\xR$ as an EPRS. Let
$\ell \R r \in \xR$ and $s = C[\ell\gamma\lnf]$ for some substitution
$\gamma$ with $\im{\gamma} \subseteq \nfbel$. Furthermore, let
$\gamma^{-\lambda}$ be the meta-substitution with respect to
$\ell \RR r$ defined from $\gamma$. By definition,
$\im{\gamma^{-\lambda}} \subseteq \nfbel$. Thus, we can use Lemmata
\ref{lem:minuslam}, \ref{lem:both} and \ref{lem:lhs} to obtain
$\ell\gamma\lnf = \ell(\gamma^{-\lambda})^\lambda\lnf =
\ell\gamma^{-\lambda}\bnf = \ell\gamma^{-\lambda}$
which means we also have $s = C[\ell\gamma^{-\lambda}]$.
We can finish the proof as in the previous case. \qedhere
\end{itemize}
\end{proof}

\section{Conclusion}

For terms in $\nfel$, \thmref{main} yields the correspondence
\[
\RRab{\beta}{!} \cdot \RRb{\xR} \cdot \RRab{\beta}{!} \cdot \RRb{\xR}
\cdot \RRab{\beta}{!} \cdots
\qquad =\qquad
\RRab{\beta}{!} \cdot \RbR \cdot \RbR \cdots
\]
of rewrite sequences. Hence, we obtain the following corollary.

\begin{corollary}
\label{cor:result}
Given a well-behaved \textup{STMRS} $\xR$,
$\RRab{\beta}{!} \cdot \RRb{\xR}$ is terminating
on terms in $\nfel$ if and only if
$\Rb{\xR}$ is terminating when viewing $\xR$ as an \textup{EPRS}.
\end{corollary}

This lays the foundation for an EPRS subcategory of the higher-order
category in TermCOMP: One simply has to make sure to only use
well-behaved STMRSs by filtering TPDB accordingly or automatically
transforming all STMRS benchmarks to well-behaved STMRSs.
Given a well-behaved STMRS, termination of the corresponding
EPRS coincides with termination of normalized rewriting \cite{M96} of
the given STMRS with respect to $\beta$-reduction.
For second-order STMRSs, meta-substitution cannot create $\beta$-redexes,
so on terms in $\nfbel$ there is even a one-to-one correspondence between
steps in a well-behaved STMRS and its corresponding EPRS.

Regarding future work, one could consider dropping the restriction
to terms in $\nfel$ in \corref{result} along the lines of
the development in \cite[Section 2.3.2]{K12}.
Furthermore, it could be worthwhile to generalize \thmref{main}
to more general subclasses of HRSs. Note that it will not
be possible to go for all HRSs 
since meta-substitution is strictly less powerful than
full higher-order substitution (see e.g.\ \exaref{stmrsVSeprs}).
Finally, by establishing the implementation result of \thmref{main}
instead of a direct proof of \corref{result}, we also provide
a solid foundation for transferring properties other than
termination from EPRSs to STMRSs.

\bibliography{references}

\end{document}